\documentclass[11pt,a4paper]{article}
\usepackage[T1]{fontenc}
\usepackage[utf8]{inputenc}
\usepackage{authblk}
\usepackage{booktabs}
\usepackage{enumitem}
\usepackage[centertags]{amsmath}
\usepackage{amsfonts,amsthm,amssymb}
\usepackage{amssymb}
\usepackage{amsmath}
\usepackage{url}
\usepackage{soul}
\usepackage{graphicx}
\usepackage{accents}
\usepackage{booktabs}
\usepackage{appendix}
\usepackage[hmargin=3.175cm,vmargin=3.175cm]{geometry}
\usepackage{etoolbox}
\usepackage{xparse}
\usepackage{enumitem}
\usepackage{authblk}
\usepackage{color}
\usepackage{soul}
\usepackage{chngcntr}
\usepackage{authblk}
\usepackage{pgf}
\usepackage{bbm}

\usepackage{multirow}
\newcommand{\ubar}[1]{\underaccent{\bar}{#1}}
\DeclareDocumentCommand{\publicBelief}{O{\mu}}{#1}
\DeclareDocumentCommand{\privateBelief}{O{p}}{#1}
\DeclareDocumentCommand{\signalupdate}{O{q}}{#1}
\DeclareDocumentCommand{\signal}{O{s}}{#1}
\DeclareDocumentCommand{\signalsSet}{O{S}}{#1}
\DeclareDocumentCommand{\state}{O{\omega}}{#1}
\DeclareDocumentCommand{\statesSet}{O{\Omega}}{#1}
\DeclareDocumentCommand{\price}{O{\tau}}{#1}
\DeclareDocumentCommand{\action}{O{a}}{#1}
\DeclareDocumentCommand{\limitParam}{O{\alpha}}{#1}
\DeclareDocumentCommand{\deterrencePrice}{O{\price} O{d}}{#1^#2}
\DeclareDocumentCommand{\LLR}{O{x}}{\log(\frac{#1}{1-#1})}
\DeclareDocumentCommand{\llr}{O{x}}{l\left(#1\right)}
\DeclareDocumentCommand{\lBound}{O{\limitParam} O{\publicBelief}}{\ubar{#1}_{#2}}
\DeclareDocumentCommand{\uBound}{O{\limitParam} O{\publicBelief}}{\bar{#1}_{#2}}
\DeclareDocumentCommand{\eqPrice}{O{\price}}{#1^{*}}

\newtheorem{theorem}{Theorem}

\theoremstyle{definition}

\newlist{secenum}{enumerate}{10}
\setlist[secenum]{label=\thesection.\arabic*,leftmargin=*}

\newtheorem{innercustomgeneric}{\customgenericname}
\providecommand{\customgenericname}{}
\newcommand{\newcustomtheorem}[2]{%
	\newenvironment{#1}[1]
	{%
		\renewcommand\customgenericname{#2}%
		\renewcommand\theinnercustomgeneric{##1}%
		\innercustomgeneric
	}
	{\endinnercustomgeneric}
}

\newcustomtheorem{customthm}{Theorem}
\newcustomtheorem{customlemma}{Lemma}
\newcommand{\blocktheorem}[1]{%
	\csletcs{old#1}{#1}% Store \begin
	\csletcs{endold#1}{end#1}% Store \end
	\RenewDocumentEnvironment{#1}{o}
	{\par\addvspace{1.5ex}
		\noindent\begin{minipage}{\textwidth}
			\IfNoValueTF{##1}
			{\csuse{old#1}}
			{\csuse{old#1}[##1]}}
		{\csuse{endold#1}
		\end{minipage}
		\par\addvspace{1.5ex}}
}

\blocktheorem{customlemma}

\definecolor{ao}{rgb}{0.0, 0.5, 0.0}

\usepackage{hyperref}
\usepackage{subcaption}

\usepackage{caption}

\usepackage{blindtext}

\usepackage[utf8]{inputenc}
\usepackage{tikz}
\usepackage{pgfplots}
\usepackage{placeins}
\usepackage{subcaption}

\usepackage{draftwatermark}
\SetWatermarkText{DRAFT}
\SetWatermarkScale{1.5}

\begin{document}

\title{A Practical Approach to Social Learning.}

\author[1]{Amir Ban\thanks{amirban@me.com and korenm@stanford.edu }}
\author[2]{Moran Koren\thanks{Research supported by the Fulbright Postdoctoral Fellowship 2019/2020}}
\affil[1]{Weizmann Institute of Science}
\affil[2]{Stanford University}

\renewcommand\Authands{, and }
\renewcommand\footnotemark{}

\date{\today}
\maketitle
\begin{abstract}
 Models of social learning feature either binary signals or abstract signal structures often deprived of micro-foundations.  Both models are limited when analyzing interim results or performing empirical analysis.  We present a method of generating signal structures which are richer than the binary model, yet are tractable enough to perform simulations and empirical analysis.  We demonstrate the method's usability by revisiting two classical papers: (1) we discuss the economic significance of unbounded signals \cite{Smith2012}; (2) we use experimental data from  \cite{Anderson1997} to perform econometric analysis. Additionally, we provide a necessary and sufficient condition for the occurrence of action cascades. 
\end{abstract}
\newpage

The literature of social learning focuses on the question ``Why do people often emulate the actions of predecessors, even when those actions contradict their own private information?''. The canonical model of Banerjee \cite{Banerjee1992}, and of Bikhchandani, Hirshleifer and Welch \cite{bikhchandani1992theory}, provided a simple and intuitive answer: ``when the history is sufficiently skewed against her private information, a rational agent's optimal action is to ignore it and follow the history.'' They do so by presenting a model in which agents, who receive a noisy private signal over an underlying state of nature, are called to act sequentially. Agent signals are binary, and their quality is commonly known. Each agent sees the actions of previous arrivals, updates her belief over the space set, using the information revealed by this action history, and chooses the action which maximizes her expected utility. Due to this elegant, yet simple, model, they show that in any game trajectory, at some point, agents will follow in the footsteps of those who precede them, even if, {\em a priori}, their signal favors the other action. These models generated great interest from both experimental economists, trying to construct information cascades in a lab (see \cite{Anderson1997}), as well as from econometricians, trying to estimate the effect of social learning on field data (see \cite{Zhang2010}). The main challenge facing those practitioners originated from the binary signal structure. In those models, a cascade emerges whenever the number of agents who took one action exceeded those who took the other by two. To circumvent this theoretical limitation, researchers resorted to attempting to estimate ``trace evidence'' for the occurrence of cascades (see \cite{Zhang2010}), assuming agents make mistakes (\cite{Anderson1997}), or assuming that agent valuation for both actions fluctuates (see \cite{Goeree2007}). All these methods violate the simple intuition mentioned above.

A second wave of research on social learning originated from Smith and S\o rensen \cite{Smith2012}. In those models, the signal structure is assumed to be abstract. Their focus is often on the asymptotic efficiency of the public belief convergence process. If the public belief converges to an interior point (as in the binary models), then \textit{an information cascade} occurs. If the convergence is to either zero or one, then \textit{learning occurs}.  Smith and S\o rensen \cite{Smith2012} contributed two major results to the discussion: (1) For learning to occur, one requires that for every history, there will always be a positive chance that the agent will choose a contrary action (a condition they call \textit{unbounded signals}), otherwise an information cascade occurs in finite time. (2) They distinguished between an information cascade (i.e., convergence of the public belief to an interior point) and an action cascade (i.e., at some point the agent chooses an action with probability one). Herrera and H\o rner \cite{herrera2012necesssary} present a condition for the occurrence of information cascades, when the domain is compact and the information structure satisfy several technical conditions.   In Section \ref{sec: condition} we show that the condition does not hold in the general case and present  a counterexample. 

 Smith and S\o rensen \cite{Smith2012}'s results generated great interest among theoreticians attempting to challenge it in various information structures (see \cite{Acemoglu2010}), agent utility functions (see \cite{Eyster2014}), or market structures (see \cite{Arieli2019}). However, the vast majority of those results focus on asymptotic analysis, and often struggle to provide economic insight on agents' short-term behavior. 

In this work we attempt to bridge this gap by presenting a method to generate signal structures which are richer than the binary signals models, yet are more tractable than the abstract signal models. Therefore our method can be used both to perform empirical analysis while maintaining the core intuition of information cascades, and to construct examples which convey economic short term behavior, thus supplementing theoretical research.  As an additional theoretic contribution, we present a necessary and sufficient condition for the occurrence of action cascades in general signal structures.

The structure of the paper is as follows. In Section \ref{sec:model} we present our model. In Section \ref{sec:SnS} we revisit the results of \cite{Smith2012} and discuss its economic significance in the short term. Section \ref{sec: condition} contains our condition for action cascades. In Section \ref{sec:econometric}, we demonstrate the applicability of our method to econometric analysis by reverse engineering the experiment of \cite{Anderson1997}. In Section \ref{sec:discussion}, we conclude.

\section{The Model}\label{sec:model}

There are two possible states of nature, $\omega \in \Omega=\{0,1\}.$ The prior probability of the realized state being $1$ is denoted by $\Pr(\omega=1)=\mu_0.$ This prior probability is commonly known, and is often dubbed the initial public belief.  
There is a countable set of agents $N.$ The set of available actions for  agent $t$ is $A_t=A\equiv\{0,1\}$ for all $t\in N.$  Agent utilities are determined by the realized state in the following manner,
$$
u_t(a_t)=\begin{cases}
1&\mbox{ if }a_t=\omega\\
-1&\mbox{ if }a_t\ne \omega
\end{cases}.
$$
Agents arrive sequentially by a predetermined order. Without loss of generality, we denote each agent by her arrival time. That is, for all $t\in N,$ we assume that agent $t$ arrives at period $t.$   Each agent receives a private signal $s_t\in S,$ where $S$ is the set of possible signals and is identical for all agents. 

We follow Banerjee \cite{Banerjee1992} and Bikhchandani et al. \cite{bikhchandani1992theory} and assume that the signal set is binary, i.e. $S=\{0,1\}.$  Let $q_t=\Pr(s_t=\omega),$ where $\omega\in\{0,1\}$ is the realized state of nature. %\footnote{Amir: Should be simply $\Pr(s_i=\omega)$? Moran: No, it is the probability that the signal is correct.... Amir: Why conditional on $\omega$? The quality does not depend on $\omega$.} denote the quality of  Agent $i's$ private signal.  
 In the canonical models mentioned above, a common assumption is that $q_t=q$ for every $t.$ That is, the quality of all agent signals is the same.    We diverge from this assumption and assume that the signal quality $q_t$ is also  private and is independently drawn from a known set $Q\subseteq [\frac{1}{2},1],$ according to some distribution  $F(\cdot).$ Note that $F(\cdot)$ is not state-dependent and is commonly known. We denote  by $f$ the corresponding density (if $Q$ is a continuous sample space or PMF otherwise). %\footnote{Note that here, the term ``agent type" refers to her signal quality. However, in \cite{Smith2012}, this term relates to the agent's preferences over the possible states. Our model therefore is a case of Smith and S\o rensen's single rational type model.}

Agents observe all actions taken previously to their arrival. Let $H_t\subseteq [0,1]^{t-1}$ denote the set of possible action histories at time $t,$ where $H_1=\{\emptyset\}.$  Let $H=\cup_{t\ge 1} H_t$ be the set of all finite histories, and $H_\infty=H\cup \{0,1\}^\infty$ be the set of all infinite histories.

 A \textit{strategy} for agent $t$ is a measurable function  $\sigma_t:H_t\times S\times Q\rightarrow\Delta(A)$ which maps every history to a decision rule.  We denote a profile of agent strategies by $\bar{\sigma}=(\sigma_t)_{t\ge 1}.$ A strategy profile $\bar\sigma,$  together with the information structure $(S,Q,F),$ and the initial public belief $\mu_0$ induces a probability distribution $P_{\bar\sigma}$ over $\Omega\times H_\infty \times S^\infty\times Q^\infty.$  We define the \textit{public belief} at time $t$, $\mu_t=P_{\bar{\sigma}}(\omega=1|h_t)$ as the probability that the state is $1$,  conditional on the realized history $h_t$ before $t$'s action.  Agents update their beliefs using Bayes' rule, and  hence the expected utility of an agent for action $a=1$  can be written as,
\begin{eqnarray}
 u_t(a_t=1|s_t=0,q_t)&=\frac{\mu_t (1-q_t)}{\mu_t (1-q_t)+(1-\mu_t)q_t}-\frac{(1-\mu_t)q_t}{\mu_t (1-q_t)+(1-\mu_t)q_t}\\
 u_t(a_t=1|s_t=1,q_t)&=\frac{\mu_t q_t}{\mu_t q_t+(1-\mu_t)(1-q_t)}-\frac{(1-\mu_t)(1-q_t)}{\mu_t q_t+(1-\mu_t)(1-q_t)}
% u_t(a_t=1|s_t=0,q_t)&=\frac{\mu_t (1-q_t)}{\mu_t (1-q_t)+(1-\mu_t)q_t}-\frac{(1-\mu_t)q_t}{\mu_t (1-q_t)+(1-\mu_t)q_t)}.
\end{eqnarray}
This agent with $s_t=0$ will play $a=1$ whenever $\mu_t>q_t.$ Similarly an agent with $s_t=1$ will play $a_t=1$ whenever $\mu_t>1-q_t.$

 As the value of $q_t$ is unknown to future arriving agents, calculating the updating rule requires some work. To do so we use the distribution of signal qualities $F$ to derive a distribution over possible agent posteriors. For every pair $s_t,q_t$ we denote by $x(s_t,q_t)=Pr(\omega=1|\mu=\frac{1}{2},q_t,s_t).$  Note that for all agents other than $t$, $x(s_t,q_t)$ is a random variable. We suppress the notation $s_t,q_t$ and let the agent {\em type} $x$ be a random variable describing agent $t$'s  posterior belief whenever $\mu_t=\frac{1}{2}$.

 Let $\bar q = \sup Q.$   Recalling the definition of quality $q_t=\Pr(s_t=\omega)$, $x$ is a random variable with support in $[1-\bar q,\bar q]$ with the following state-conditional densities $g_\omega(\cdot)$ for state $\omega\in\{0,1\},$
\begin{eqnarray}
&g_1(x)=\begin{cases}
\label{x1}
xf(x)&\mbox{ if }x\in[\frac{1}{2},\bar{q}]\\
xf(1-x)&\mbox{ if }x\in[1-\bar{q},\frac{1}{2}]
\end{cases}\\
&g_0(x)=\begin{cases}
\label{x0}
(1-x)f(x)&\mbox{ if }x\in[\frac{1}{2},\bar{q}]\\
(1-x)f(1-x)&\mbox{ if }x\in[1-\bar{q},\frac{1}{2}]
\end{cases}
\end{eqnarray}
and define $G_\omega(x)=\int_{0}^{x}g_\omega(z)dz,$ as the CDF of the state conditional distributions.

Note that for any quality distribution $F,$ the ratio $\frac{g_1(x)}{g_0(x)}$ equals $\frac{x}{1-x},$  thus increasing in $x.$ I.e., the type distribution exhibits the Monotone (increasing) Likelihood Ratio Property (MLRP).  By Bayes' rule, agent $t$'s private posterior $\mu := \Pr(\omega=1| h_t, x_t)$ satisfies $$\frac{\mu}{1 - \mu} = \frac{\mu_{t}}{1-\mu_{t}} \frac{g_1(x_t)}{g_0(x_t)} = \frac{\mu_{t}}{1-\mu_{t}} \frac{x_t}{1 - x_t}$$ and therefore the optimal strategy of agent $t$ is a threshold strategy. That is, for every $\mu_t,$ there exists $\tilde x(\mu_t)\in [1-\bar q,\bar q]$ such that whenever $x<\tilde x(\mu),$ $a_t=0$ and whenever $x > \tilde x(\mu)$, $a_t=1.$ We denote $\mu_t^+:=\Pr(\omega=1|h_t,a_t=1)$ and $\mu_t^-:=\Pr(\omega=1|h_t,a_t=0),$ and formulate the updating rules as follows
\begin{eqnarray}
&\frac{\mu^+_{t+1}}{1-\mu^+_{t+1}}=\frac{\mu_{t}}{1-\mu_{t}}\frac{1-G_1(\tilde x(\mu_t))}{1-G_0(\tilde x(\mu_t))}\label{eq:mu+}\\
&\frac{\mu^-_{t+1}}{1-\mu^-_{t+1}}=\frac{\mu_{t}}{1-\mu_{t}}\frac{G_1(\tilde x(\mu_t))}{G_0(\tilde x(\mu_t))} \label{eq:mu-}
\end{eqnarray}
Unlike the abstract signal structure of Smith and S\o rensen \cite{Smith2012}, for a given $F(\cdot),\mu_0,$ and $Q,$ equations \ref{eq:mu+} and \ref{eq:mu-} can be used to calculate the updated public belief following any finite length history $h_t=\{a_1,a_2,\dots,a_{t-1}\}.$  

In addition, Agent $t$, with type $x,$ will play $a_t=1$ whenever $$\frac{x}{1-x}>\frac{1-\mu_t}{\mu_t}$$ An up-cascade occurs whenever $\mu_t>\bar{q}$ and a down-cascade occurs whenever $\mu_t<1-\bar{q}.$ The cascade regions, unsurprisingly, are identical to those of the binary model  of Banerjee \cite{Banerjee1992} and Bikhchandani et. al. \cite{bikhchandani1992theory}.  The difference is in the time of convergence, i.e. in the number of consecutive actions required to induce a cascade. In the following examples we examine this aspect using several distribution families. 

\subsection{Example: Uniform distribution}\label{sec:UD}

To illustrate the uses of the method described above, we introduce a simple example.  Assume that $q_t\sim U[\frac{1}{2},\bar q]$ for every $t$. 

By equations \eqref{x1} and \eqref{x0} we can calculate the following distributions,

$$
G_1(x)=\int_{1-\bar q}^{x} r\frac{1}{\bar q -\frac{1}{2}}dr = \begin{cases}
0&\mbox{ if }x<1-\bar q\\
\frac{x^2-(1-\bar{q})^2}{2\bar q-1}&\mbox{ if }x\in[1-\bar q,\bar q]\\
1&\mbox{ if }x>\bar q
\end{cases}
$$
and
$$
G_0(x)=\int_{1-\bar q}^{x} (1-r)\frac{1}{\bar q -\frac{1}{2}}dr = \begin{cases}
0&\mbox{ if }x<1-\bar q\\
\frac{\bar{q}^2-(1-x)^2}{2\bar q-1}&\mbox{ if }x\in[1-\bar q,\bar q]\\
1&\mbox{ if }x>\bar q.
\end{cases}
$$

Assume that $\mu_0=\frac{1}{2}$, $h_3=\{1,1\} $ and $\bar q =\frac{2}{3}$. We can calculate agent thresholds in the following way:
\begin{eqnarray*}
&\frac{x_1}{1-x_1}=1\Rightarrow x_1=\frac{1}{2}\Rightarrow \mu^+=\frac{7}{12}\Rightarrow\\ 
& \Rightarrow x_2=\frac{5}{12}\Rightarrow \frac{\mu^{++}}{1-\mu^{++}}=\frac{1-G_1(\frac{1}{2})}{1-G_0(\frac{1}{2})}\frac{1-G_1(\frac{5}{12})}{1-G_0(\frac{5}{12})}\Rightarrow \mu^{++}=\frac{91}{146}\approx 0.623  \\
& \Rightarrow x_3=\frac{55}{146}>\frac{1}{3}
\end{eqnarray*}

As shown in \cite{Banerjee1992,bikhchandani1992theory}, in the classic binary signal model,  the public belief following every history is determined by the initial public belief and the difference between the number of $a=1$ and $a=0$ taken. Whenever this difference  is greater than two, agent actions are no longer informative, thus a cascade occurs and $\mu_{t+1}=\mu_t.$  Note that here, unlike in the case of the classical model with binary signals, after $h_3=\{1,1\},$ agent $3$ still plays $a=0$ with positive probability. This attribute makes possible more direct methods of empirical analysis (as we show in Section \ref{sec:econometric}), but also allows us to gain further insight into the interim periods of the observational learning process, as we show in the following section.

\section{The Economic Significance of Unbounded Signals.}\label{sec:SnS}

In their seminal work, Smith and S\o rensen \cite{Smith2012} generalized the game information structure from one in which the signals are binary to a structure in which abstract signals are drawn from one of two  state-dependent distributions. This extension provided important insights into the governing forces of herding. Their first result stated that when the initial belief is not in the cascade region and signals are not discrete, information cascades do not occur as the public belief never crosses into the cascade region, but converges to its border. Their second result stated that despite the scarcity of information cascades, the history of actions will ``settle'' on an alternative. They identified a necessary and sufficient condition under which the public belief converges to the true state of the world.

Smith and S\o rensen classified the game information structure into bounded or unbounded beliefs. When signals are unbounded, at any public belief, and after any history, there is always a positive chance that an agent will receive a signal strong enough to induce a contrary action.  In our model, this translates to $\bar{q} =1.$ In this section we revisit their classic results using the  example from Section \ref{sec:UD}.

In the table below we calculated the probability of a contrary action following a history with an initial belief of $\frac{1}{2}$ and a sequence of  1,2,4, and 8 consecutive actions for several values of maximal signal quality $\bar{q}$. One can see that when signals are unbounded, the probability of a contrary action remains significant even after 8 consecutive actions. In addition, note that when signals are bounded, the public belief stabilizes rapidly to the border of the cascade region, yet it never crosses it. In addition, note that the effect signal boundedness has on the  process of social learning can best be witnessed when the sequence of consecutive actions is sufficiently long. For example, when the history is $\{0,0\}$, the probability of $a=1$ is roughly the same for all levels of $\bar{q}$. However, when $h=\{0,0,0,0,0,0,0,0\},$ the probability of $a=1$, is greater by an order of magnitude, from $0.0021$ when $\bar q=0.55$ to $0.0556$ when $\bar q=1.$

% \usepackage{booktabs}

\begin{table}[!htbp]
	\centering\resizebox{\columnwidth}{!}{%
	\begin{tabular}{@{}|c|cc|cc|cc|cc|@{}}
		\toprule
		\multirow{2}{*}{\begin{tabular}[c]{@{}c@{}}$\mu_0=\frac{1}{2}$ \\  $h$\end{tabular}}& \multicolumn{2}{c|}{$\bar q=0.55$}  & \multicolumn{2}{c|}{$\bar q =0.66$}   & \multicolumn{2}{c|}{$\bar q = 0.77$}  & \multicolumn{2}{c|}{$\bar q = 1$}     \\ \cmidrule(l){2-9} 
		 & $Pr(a=1|h)$ & $\mu_{t+1}(h)$ & $Pr(a=1|h)$ & $\mu_{t+1}(h)$ & $Pr(a=1|h)$ & $\mu_{t+1}(h)$ & $Pr(a=1|h)$ & $\mu_{t+1}(h)$ \\ \midrule
		\{0\}                                                                             & 0.25        & 0.475          & 0.25        & 0.42           & 0.25        & 0.365          & 0.25        & 0.25           \\
		\{0,0\}                                                                            & 0.125       & 0.4625         & 0.125        & 0.382          & 0.138       & 0.304          & 0.167       & 0.167          \\
		\{0,0,0,0\}                                                                          & 0.032       & 0.454          & 0.0374      & 0.352          & 0.05        & 0.257          & 0.1         & 0.1            \\
		\{0,0,0,0,0,0,0,0\}                                                                      & 0.0021      & 0.450          & 0.0034      & 0.341          & 0.008       & 0.234          & 0.0556      & 0.0556         \\ \bottomrule
	\end{tabular}
}
\caption{The probability of a contrary action and convergence of public belief when the maximal signal quality varies, in the family of  information structures presented in Section \ref{sec:UD}.}
\end{table}
\FloatBarrier
\subsection{Occurrence of cascades with continuous signals.}\label{sec: condition}

When generalizing the binary model to one with abstract signals, Smith and S\o rensen distinguished between two types of cascades: (1) a public belief cascade, i.e., a case in which the public belief converges to an interior point, which occurs whenever signals are bounded, and (2) an  action cascade, which describes a case where the agent actions converge to a single one, occuring whenever the public belief crosses into the cascade region. For the latter Smith and S\o rensen \cite{Smith2012} stated that it occurs with positive probability when the distribution tails contains atoms. In \cite{herrera2012necesssary},   Herrera and H\"{o}rner study the existence of action cascades in models with continuous sample spaces and prove that, under some technical conditions,\footnote{In \cite{herrera2012necesssary}, they require that the information structure satisfies MLRP, that the signal space is compact, and that the density for every $x$ is bounded away from zero.}  when signals are continuous, an action cascade occurs if and only if the information structure does not exhibit the \textit{increasing hazard ratio property (IHRP).} \footnote{By \cite{herrera2012necesssary}, (strict) IHRP holds if the following mapping is increasing with signal $x,$ 
$H(x)=\frac{1-G_0(x)}{1-G_1(x)}\frac{g_1(x)}{g_0(x)}.$}
 In this section we  show that the condition provided in \cite{herrera2012necesssary} does not hold in the general case. We do so by constructing a counter-example. In addition, in Theorem \ref{thm:cascade_cond} we provide an alternative necessary and sufficient condition for the occurrence of action cascades in the general case.

To study the occurrence of action cascades we study a generalized version of the example presented in Section \ref{sec:UD}. In this family of information structures we assume that the agent signal qualities are distributed uniformly between $[\ubar{q},\bar q]$ for some $1\ge\bar q \ge\ubar q\ge\frac{1}{2}.$ When $\ubar q=\frac{1}{2}$ we get the example from Section \ref{sec:UD}, and when $\bar q=\ubar q$ we get the binary signal model of \cite{Banerjee1992,bikhchandani1992theory}.  Using Python, we  calculated the number of consecutive $a=1$ actions required for action cascades, starting at $\mu_0=\ubar{q}$ when $\bar q =0.8$. The results, depending on the of value of $\ubar q$, are plotted in Figure \ref{fig:action cascades}.
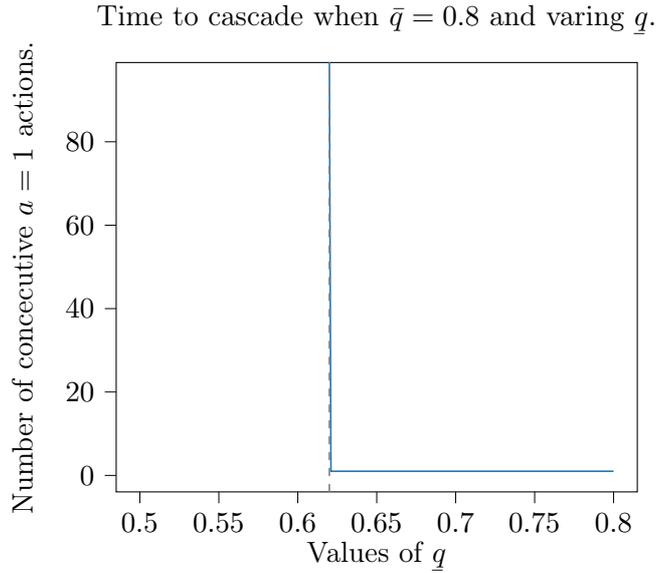
\begin{figure}[!htbp]
	\centering
	% This file was created by tikzplotlib v0.8.7.
\begin{tikzpicture}

\definecolor{color0}{rgb}{0.12156862745098,0.466666666666667,0.705882352941177}

\begin{axis}[
tick align=outside,
tick pos=left,
title={Time to cascade when \(\displaystyle \bar q =0.8\) and varing \(\displaystyle \underaccent{\bar}{ q} .\)},
x grid style={white!69.01960784313725!black},
xlabel={Values of \(\displaystyle \underaccent{\bar}{ q}\)},
xmin=0.485, xmax=0.815,
xtick style={color=black},
y grid style={white!69.01960784313725!black},
ylabel={Number of concecutive \(\displaystyle a=1\) actions.},
ymin=-3.95, ymax=99,
ytick style={color=black}
]
\addplot [semithick, color0]
table {%
0.5 100
0.501 100
0.502 100
0.503 100
0.504 100
0.505 100
0.506 100
0.507 100
0.508 100
0.509 100
0.51 100
0.511 100
0.512 100
0.513 100
0.514 100
0.515 100
0.516 100
0.517 100
0.518 100
0.519 100
0.52 100
0.521 100
0.522 100
0.523 100
0.524 100
0.525 100
0.526 100
0.527 100
0.528 100
0.529 100
0.53 100
0.531 100
0.532 100
0.533 100
0.534 100
0.535 100
0.536 100
0.537 100
0.538 100
0.539 100
0.54 100
0.541 100
0.542 100
0.543 100
0.544 100
0.545 100
0.546 100
0.547 100
0.548 100
0.549 100
0.55 100
0.551 100
0.552 100
0.553 100
0.554 100
0.555 100
0.556 100
0.557 100
0.558 100
0.559 100
0.56 100
0.561 100
0.562 100
0.563 100
0.564 100
0.565 100
0.566 100
0.567 100
0.568 100
0.569 100
0.57 100
0.571 100
0.572 100
0.573 100
0.574 100
0.575 100
0.576 100
0.577 100
0.578 100
0.579 100
0.58 100
0.581 100
0.582 100
0.583 100
0.584 100
0.585 100
0.586 100
0.587 100
0.588 100
0.589 100
0.59 100
0.591 100
0.592 100
0.593 100
0.594 100
0.595 100
0.596 100
0.597 100
0.598 100
0.599 100
0.6 100
0.601 100
0.602 100
0.603 100
0.604 100
0.605 100
0.606 100
0.607 100
0.608 100
0.609 100
0.61 100
0.611 100
0.612 100
0.613 100
0.614 100
0.615 100
0.616 100
0.617 100
0.618 100
0.619 100
0.62 100
0.621 1
0.622 1
0.623 1
0.624 1
0.625 1
0.626 1
0.627 1
0.628 1
0.629 1
0.63 1
0.631 1
0.632 1
0.633 1
0.634 1
0.635 1
0.636 1
0.637 1
0.638 1
0.639 1
0.64 1
0.641 1
0.642 1
0.643 1
0.644 1
0.645 1
0.646 1
0.647 1
0.648 1
0.649 1
0.65 1
0.651 1
0.652 1
0.653 1
0.654 1
0.655 1
0.656 1
0.657 1
0.658 1
0.659 1
0.66 1
0.661 1
0.662 1
0.663 1
0.664 1
0.665 1
0.666 1
0.667 1
0.668 1
0.669 1
0.67 1
0.671 1
0.672 1
0.673 1
0.674 1
0.675 1
0.676 1
0.677 1
0.678 1
0.679 1
0.68 1
0.681 1
0.682 1
0.683 1
0.684 1
0.685 1
0.686 1
0.687 1
0.688 1
0.689 1
0.69 1
0.691 1
0.692 1
0.693 1
0.694 1
0.695 1
0.696 1
0.697 1
0.698 1
0.699 1
0.7 1
0.701 1
0.702 1
0.703 1
0.704 1
0.705 1
0.706 1
0.707 1
0.708 1
0.709 1
0.71 1
0.711 1
0.712 1
0.713 1
0.714 1
0.715 1
0.716 1
0.717 1
0.718 1
0.719 1
0.72 1
0.721 1
0.722 1
0.723 1
0.724 1
0.725 1
0.726 1
0.727 1
0.728 1
0.729 1
0.73 1
0.731 1
0.732 1
0.733 1
0.734 1
0.735 1
0.736 1
0.737 1
0.738 1
0.739 1
0.74 1
0.741 1
0.742 1
0.743 1
0.744 1
0.745 1
0.746 1
0.747 1
0.748 1
0.749 1
0.75 1
0.751 1
0.752 1
0.753 1
0.754 1
0.755 1
0.756 1
0.757 1
0.758 1
0.759 1
0.76 1
0.761 1
0.762 1
0.763 1
0.764 1
0.765 1
0.766 1
0.767 1
0.768 1
0.769 1
0.77 1
0.771 1
0.772 1
0.773 1
0.774 1
0.775 1
0.776 1
0.777 1
0.778 1
0.779 1
0.78 1
0.781 1
0.782 1
0.783 1
0.784 1
0.785 1
0.786 1
0.787 1
0.788 1
0.789 1
0.79 1
0.791 1
0.792 1
0.793 1
0.794 1
0.795 1
0.796 1
0.797 1
0.798 1
0.799 1
0.8 1
};
\addplot [semithick, white!50.19607843137255!black, dashed]
table {%
0.62 -3.95
0.62 104.95
};
\end{axis}

\end{tikzpicture}
	\caption{The occurrence of action cascades when signal qualities are drawn from a uniform distribution over $[\underline{q},\bar{q}]$ (for  $\bar{q}=0.8$ and varied $\underline{q}$).}\label{fig:action cascades}
	\end{figure}
\FloatBarrier
By Figure \ref{fig:action cascades}, one can see that action cascades do not occur for $\ubar q<0.620$ and occur after the history $h_t=\{1\}, \mu_0=\ubar{q},$  whenever $\ubar q\ge 0.620.$  From this one can rule out atoms in the distribution tails as the cause for cascades as our signals are continuous yet action cascades do occur whenever  $\ubar q>0.620.$ In addition, the IHRP condition of Herrera and H\"{o}rner also seems inaccurate as our results demonstrate a counter-example. E.g, when $\ubar q=0.56,$ the information structure violates the IHRP condition (see the orange curve in Figure \ref{fig:IHRP}), yet action cascades do not occur.\footnote{The attentive reader may think that this contradiction is due to the fact that our example violates Herrera and H\"{o}rner's assumption of a compact domain. However, in Appendix \ref{sec:ihrp_compact_ex}, we provide an example with a compact domain, which violates IHRP, yet no action cascade occurs.} In the following section we provide a  necessary and sufficient condition for  action cascades.
\begin{figure}
	\centering
	% This file was created by tikzplotlib v0.8.7.
\begin{tikzpicture}

\definecolor{color0}{rgb}{0.12156862745098,0.466666666666667,0.705882352941177}
\definecolor{color1}{rgb}{1,0.498039215686275,0.0549019607843137}
\definecolor{color2}{rgb}{0.172549019607843,0.627450980392157,0.172549019607843}
\definecolor{color3}{rgb}{0.83921568627451,0.152941176470588,0.156862745098039}

\begin{axis}[
legend cell align={left},
legend style={fill opacity=0.8, draw opacity=1, text opacity=1, at={(0.03,0.97)}, anchor=north west, draw=white!80.0!black},
tick align=outside,
tick pos=both,
x grid style={white!69.01960784313725!black},
xlabel={\(\displaystyle $x$\)},
xmin=0.181, xmax=0.819000000000001,
xtick style={color=black},
y grid style={white!69.01960784313725!black},
ylabel={\(\displaystyle \frac{1-G_0(x)}{1-G_1(x)}\frac{g_1(x)}{g_0(x)}\)},
ymin=0.146708698623928, ymax=1.00925768558158,
ytick style={color=black}
]
\addplot [semithick, color0,unbounded coords=jump]
table {%
0.21 0.260558967289134
0.22 0.270990447461036
0.23 0.281301223048796
0.24 0.291497975708502
0.25 0.301587301587301
0.26 0.311575726670066
0.27 0.32146972218666
0.28 0.331275720164609
0.29 0.341000129215661
0.3 0.350649350649351
0.31 0.360229795012404
0.32 0.369747899159664
0.33 0.379210143970413
0.34 0.388623072833599
0.35 0.397993311036789
0.36 0.407327586206897
0.37 0.416632749966083
0.38 0.425915800984144
0.39 0.435183909629426
0.4 0.444444444444444
0.41 0.453705000700378
0.42 0.462973431317128
0.43 0.472257880473542
0.44 0.481566820276498
0.45 0.490909090909091
0.46 0.500293944738389
0.47 0.509731094933888
0.48 0.519230769230769
0.49 0.528803769569844
0.5 0.538461538461539
0.51 0.548216233058109
0.52 0.558080808080808
0.53 0.568069108942569
0.54 0.578195976638546
0.55 0.588477366255144
0.56 0.598930481283423
0.57 0.609573926328297
0.58 0.620427881297447
0.59 0.631514300754518
0.6 0.642857142857143
0.61 0.654482633206038
0.62 0.666419570051891
0.63 0.678699678699679
0.64 0.691358024691358
0.65 0.704433497536946
0.66 0.717969379532636
0.67 0.732014017728303
0.68 0.746621621621622
0.69 0.761853214981598
0.7 0.777777777777778
0.71 0.794473624115095
0.72 0.81203007518797
0.73 0.830549503752118
0.74 0.85014985014985
0.75 0.870967741935484
0.76 0.893162393162393
0.77 0.916920520631403
0.78 0.942462600690448
0.79 0.970050913447137
};
\addlegendentry{$\underbar{q}=0.5$}
\addplot [semithick, color1,unbounded coords=jump]
table {%
0.21 0.259231675625929
0.22 0.268175985733238
0.23 0.276831068738583
0.24 0.285193953742488
0.25 0.293260473588342
0.26 0.301025163094129
0.27 0.308481140255467
0.28 0.315619967793881
0.29 0.322431491910859
0.3 0.32890365448505
0.31 0.335022274189633
0.32 0.340770791075051
0.33 0.346129968014464
0.34 0.35107754098852
0.35 0.355587808417997
0.36 0.359631147540983
0.37 0.363173443048149
0.38 0.366175409667113
0.39 0.368591785900481
0.4 0.37037037037037
0.41 0.371450864796913
0.42 0.371763477981614
0.43 0.371227232509803
0.44 nan
0.45 nan
0.46 nan
0.47 nan
0.48 nan
0.49 nan
0.5 nan
0.51 nan
0.52 nan
0.53 nan
0.54 nan
0.55 nan
0.56 0.598930481283423
0.57 0.609573926328297
0.58 0.620427881297447
0.59 0.631514300754519
0.6 0.642857142857143
0.61 0.654482633206038
0.62 0.666419570051891
0.63 0.678699678699679
0.64 0.691358024691358
0.65 0.704433497536946
0.66 0.717969379532635
0.67 0.732014017728304
0.68 0.746621621621622
0.69 0.761853214981598
0.7 0.777777777777778
0.71 0.794473624115094
0.72 0.81203007518797
0.73 0.830549503752119
0.74 0.85014985014985
0.75 0.870967741935484
0.76 0.893162393162394
0.77 0.916920520631407
0.78 0.942462600690449
0.79 0.970050913447145
};
\addlegendentry{$\underbar{q}=0.56$}
\addplot [semithick, color2,unbounded coords=jump]
table {%
0.21 0.252526011343575
0.22 0.253797440325938
0.23 0.2537185403764
0.24 0.252177205603938
0.25 0.24904214559387
0.26 0.244159413650939
0.27 0.237348114536687
0.28 0.228395061728395
0.29 0.21704807713047
0.3 0.203007518796992
0.31 0.185915470758367
0.32 nan
0.33 nan
0.34 nan
0.35 nan
0.36 nan
0.37 nan
0.38 nan
0.39 nan
0.4 nan
0.41 nan
0.42 nan
0.43 nan
0.44 nan
0.45 nan
0.46 nan
0.47 nan
0.48 nan
0.49 nan
0.5 nan
0.51 nan
0.52 nan
0.53 nan
0.54 nan
0.55 nan
0.56 nan
0.57 nan
0.58 nan
0.59 nan
0.6 nan
0.61 nan
0.62 nan
0.63 nan
0.64 nan
0.65 nan
0.66 nan
0.67 nan
0.68 0.746621621621622
0.69 0.761853214981598
0.7 0.777777777777778
0.71 0.794473624115095
0.72 0.81203007518797
0.73 0.830549503752118
0.74 0.85014985014985
0.75 0.870967741935484
0.76 0.893162393162392
0.77 0.916920520631405
0.78 0.942462600690449
0.79 0.970050913447137
};
\addlegendentry{$\underbar{q}=0.6$}
\addlegendentry{$\underbar{q}=0.68$}
\addplot [semithick, color3, ,unbounded coords=jump]
table {%
	0.21 0.257899802085298
	0.22 0.265341400172861
	0.23 0.272311258878037
	0.24 0.278793217353006
	0.25 0.28476821192053
	0.26 0.290213951868087
	0.27 0.295104540213372
	0.28 0.299410029498525
	0.29 0.303095900494283
	0.3 0.306122448979592
	0.31 0.308444062353819
	0.32 0.310008363534987
	0.33 0.310755194129023
	0.34 0.310615401860778
	0.35 0.309509388249546
	0.36 0.307345360824742
	0.37 0.304017218896796
	0.38 0.299401981753896
	0.39 0.293356641365684
	0.4 nan
	0.41 nan
	0.42 nan
	0.43 nan
	0.44 nan
	0.45 nan
	0.46 nan
	0.47 nan
	0.48 nan
	0.49 nan
	0.5 nan
	0.51 nan
	0.52 nan
	0.53 nan
	0.54 nan
	0.55 nan
	0.56 nan
	0.57 nan
	0.58 nan
	0.59 nan
	0.6 0.642857142857143
	0.61 0.654482633206038
	0.62 0.66641957005189
	0.63 0.678699678699679
	0.64 0.691358024691358
	0.65 0.704433497536946
	0.66 0.717969379532636
	0.67 0.732014017728304
	0.68 0.746621621621622
	0.69 0.761853214981598
	0.7 0.777777777777778
	0.71 0.794473624115094
	0.72 0.812030075187971
	0.73 0.830549503752119
	0.74 0.850149850149851
	0.75 0.870967741935484
	0.76 0.893162393162393
	0.77 0.916920520631407
	0.78 0.942462600690445
	0.79 0.970050913447138
};
\end{axis}

\end{tikzpicture}
	\caption{Hazard ratio for  $\bar{q}=0.8$ and various $\underbar{q}$.}\label{fig:IHRP}
\end{figure}
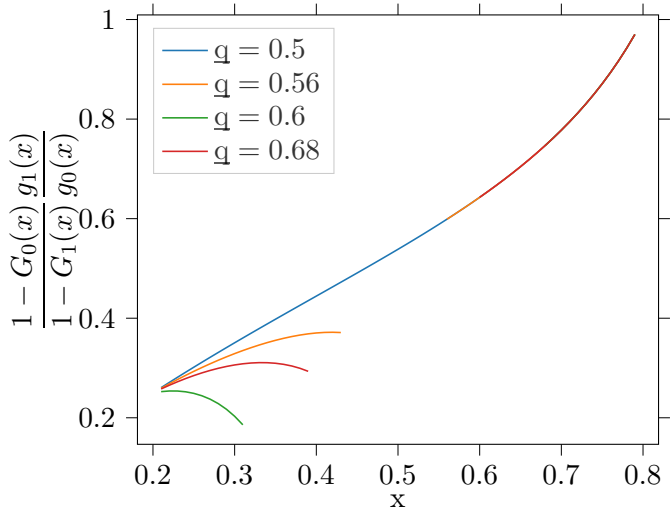

\FloatBarrier
\subsection{A necessary and sufficient  condition for  action cascades.}

In the following theorem we provide the condition for action cascades in the general MLRP case (not only in the family of information structures described in our model). We therefore slightly alter our notation.  Let $X$ denote an abstract set of signals. At state $\omega\in \{0,1\},$ agent signals are independently drawn from a state-dependent distribution $F_\omega.$ We assume that no signal perfectly reveals the state, i.e. that $F_0,F_1$ are mutually absolutely continuous with respects to each other. We denote the Radon-Nikodym derivative of $F_\omega$ with respect to the probability measure $\frac{F_0+F_1}{2}$ by $f_\omega.$\footnote{ Note that when $X$ is finite $f_\omega$ is the probability mass function of $F_\omega$ and when $X$ is continuous, $f_\omega$ is simply its density} 

 We denote the signal structure bounds as follows, $$\ubar{x}=\arg\min \frac{f_1(x)}{f_0(x)+f_1(x)},\bar{x}=\arg\max \frac{f_1(x)}{f_0(x)+f_1(x)}.$$ We assume that $\ubar{x},\bar{x}\in X$ and that the information structure exhibits the monotone likelihood property, i.e. $\frac{f_1(x)}{f_0(x)}$ increases with $x.$ 
 For example, in the information structures generated by our example from Section \ref{sec:SnS}, $X=[1-\bar q,1-\ubar{q}]\cup[\ubar{q},\bar q]$, $\bar{x} = \bar q$ and $\ubar{x} = 1-\bar q.$ Hereafter, for readability we slightly abuse notation, and denote the state conditional CDF by $F_\omega(x)$ and the state conditional PDF (or PMF if discrete) by $f_\omega(x).$ 
 
In the following theorem we provide a  necessary and sufficient  condition for  action cascades.

\begin{theorem}\label{thm:cascade_cond}
	Let $X \subseteq [\ubar{x}, \bar{x}]$ be defined as above. If the distributions $F_0,F_1$ exhibit MLRP, and  the prior public belief $\mu := \Pr(\omega=1)$ is such that an action cascade is not in progress, then
	
	 a necessary condition, and a sufficient one when strict, that an action up-cascade can start is that $\exists x\in X$ for which,
	\begin{equation}\label{eq:condition_1}
\frac{f_1(\ubar{x})}{f_0(\ubar{x})}\ge\frac{f_1(x)}{f_0(x)}\frac{1-F_0(x)}{1-F_1(x)}
\end{equation}
		and a necessary condition, and a sufficient one when strict, that an action down-cascade can start is that $\exists x\in X$ for which
	\begin{equation}\label{eq:condition_12}
	\frac{f_1(\bar{x})}{f_0(\bar{x})}\le \frac{f_1(x)}{f_0(x)}\frac{F_0(x)}{F_1(x)}
	\end{equation}
\end{theorem}
\begin{proof}	
The agent's expected utility for each action is
\begin{eqnarray}
&u(a=1|x,\mu)=\frac{f_1(x)\mu}{f_1(x)\mu+f_0(x)(1-\mu)}-\frac{f_0(x)\mu}{f_0(x)\mu+f_1(x)(1-\mu)}\\
&u(a=0|x,\mu)=-\frac{f_1(x)\mu}{f_1(x)\mu+f_0(x)(1-\mu)}+\frac{f_0(x)\mu}{f_0(x)\mu+f_1(x)(1-\mu)}
\end{eqnarray}

Thus, the agent chooses $a=1$ whenever $\frac{f_1(x)}{f_0(x)}>\frac{1-\mu}{\mu}$, $a=0$ whenever $\frac{f_1(x)}{f_0(x)}<\frac{1-\mu}{\mu}$, and is indifferent otherwise.

Therefore, whenever $\frac{f_1(\ubar{x})}{f_0(\ubar{x})}>\frac{1-\mu}{\mu}$, an up-cascade is in progress, and conversely, if an up-cascade is in progress, $\frac{f_1(\ubar{x})}{f_0(\ubar{x})}\geq\frac{1-\mu}{\mu}$. Whenever $\frac{f_1(\bar x)}{f_0(\bar x)}<\frac{1-\mu}{\mu}$, a down-cascade is in progress, and conversely, if a down-cascade is in progress, $\frac{f_1(\ubar{x})}{f_0(\ubar{x})}\leq\frac{1-\mu}{\mu}$.

We denote the \textit{threshold signal} by $\tilde{x}(\mu)=\sup\{x\in X| \frac{f_1(x)}{f_0(x)}\le\frac{1-\mu}{\mu}\}.$ Following action $a,$ other agents will  update their belief $\mu$ in the following manner.
\begin{eqnarray}
\frac{1-\mu^+}{\mu^+}&=\frac{1-\mu}{\mu}\frac{1-F_0(\tilde{x}(\mu))}{1-F_1(\tilde{x}(\mu))}\\
\frac{1-\mu^-}{\mu^-}&=\frac{1-\mu}{\mu}\frac{F_0(\tilde{x}(\mu))}{F_1(\tilde{x}(\mu))}
\end{eqnarray}
 where $\mu^+,\mu^-$ denote the updated public belief following actions $a=1,a=0$ respectively.

So, if action $a=1$ starts an up-cascade, we must have
%Assume that an information cascade occurs, i.e. there exists a pair $x,\mu$ such that, $\frac{f_1(x)}{f_0(x)}>\frac{1-\mu}{\mu}>\frac{f_1(\ubar{x})}{f_0(\ubar{x})}$ and $\frac{f_1(\ubar{x})}{f_0(\ubar{x})}>\frac{1-\mu^+}{\mu^+}.$  That is,
$$
\frac{f_1(\ubar{x})}{f_0(\ubar{x})}\geq\frac{1-\mu^+}{\mu^+}=\frac{1-\mu}{\mu}\frac{1-F_0(\tilde{x}(\mu))}{1-F_1(\tilde{x}(\mu))} \geq \frac{f_1(\tilde{x}(\mu))}{f_0(\tilde{x}(\mu))}\frac{1-F_0(\tilde{x}(\mu))}{1-F_1(\tilde{x}(\mu))}.
$$ which is \eqref{eq:condition_1} with $\tilde{x}(\mu)$ substituted for $x$.
For sufficiency note that when  this inequality is strict, and the public belief is $\frac{f_0(\tilde{x}(\mu))}{f_0(\tilde{x}(\mu)) + f_1(\tilde{x}(\mu))}$, a cascade starts after the next action, if it is $a=1.$ 
 The second part of the theorem is due to symmetric considerations.
\end{proof} 
In our example, equation \eqref{eq:condition_1} holds whenever there exist $x\in[1-\bar{q},1-\ubar{q}]\cup[\ubar{q},\bar{q}]$ for which the following holds,
$$
\frac{1-x}{x}\frac{2(\bar{q}-\ubar{q})-x^2+(1-\bar{q})^2}{2(\bar{q}-\ubar{q})-\bar{q}^2+(1-x)^2}\ge \frac{1-\bar{q}}{\bar{q}}.
$$
 When $\bar q=0.8$, by Theorem \ref{thm:cascade_cond}, action cascades occur whenever $\ubar{q}\ge 0.620$.

Furthermore, when the signal distribution is discrete at $\ubar{x},$ an $a=1$ cascade may occur, and when it is discrete around $\bar{x}$ an $a=0$ action cascade may occur, as there is positive probability for such a posterior. To see this note that when signals are binary with quality $q,$ equation \eqref{eq:condition_1}, for $s_t=0$ can be written as,
$$
\frac{1-q}{q}\frac{1-q}{q}\le\frac{1-q}{q}.
$$
This inequality holds for every $q>0.5.$

 In addition, as shown by Herrera and H\o rner \cite{herrera2012necesssary} for continuous distributions over a compact domain, and under some technical conditions (specifically, the density for every $x\in X$ is above some positive constant), the existence of a solution to equation \eqref{eq:condition_1} is determined by whether or not the information structure exhibits IHRP (see Proposition 2 in   \cite{herrera2012necesssary}).

\section{Econometric Estimation}\label{sec:econometric}

In order to demonstrate the applicability of our method for econometric estimation we revisit the experimental work of Anderson and Holt \cite{Anderson1997}. In \cite{Anderson1997} the authors conducted an experiment, designed to generate information cascades in a controlled environment. In the experiment we revisit,  each subject received a private draw from an urn which contained either two `{\em a}' labeled balls and one `{\em b}' labeled ball or two `{\em b}' labeled balls and one `{\em a}' labeled ball. Following this process, subjects were sequentially asked to declare out loud whether they think the urn is a majority-`{\em a}' urn or a  majority-`{\em b}' urn. This process yielded 90 sequences of declarations.\footnote{The experimental data for Anderson and Holt's symmetric experiment is available at \url{http://www.people.virginia.edu/~cah2k/casdata.pdf}.}  The authors then conducted an econometric analysis to prove the existence of information cascades which was constructed on the binary signal model with a known probability of agents' mistakes.

Our model assumes that agents vary in their ability to trust the informational value of their private signal. Under this assumption we attempt to ``reverse engineer'' the signal quality in the aforementioned experiment. We base our analysis on the signal structure presented in Section \ref{sec:UD} and use Hansen's General Method of Moments (GMM) estimation procedure \cite{Hansen1982}, which have been shown to be consistent, efficient, and asymptotically normal estimators,  to estimate the subject's signal quality.\footnote{We maintain the required assumption of a weakly stationary ergodic stochastic process by selecting our moment conditions which are described bellow.}

In GMM estimation, the goal is to find the model parameters for which the distance between a vector of model moments and those of the data is minimized. In social learning models, a natural choice for the model moments is $\Pr(a_t=1|h_t)$ for some subset of histories $h_t\in A\subseteq H_t.$  i.e., for every $h\in H,$ we define the moment conditions as,
$$
\Pr(a_t=1|h,q)-\sum a_t \mathbbm{1}_{h_t=h}.
$$
where $\mathbbm{1}_{h_t=h}$ is an indicator function which return 1 if the history up to action $a_t$ equal to $h$ and zero otherwise and $q$ is the parameter which defines the distribution as introduced in Section \ref{sec:UD}.

We denote the proportion of $a=1$ actions taken following the history $h_t$ in the data by $\phi_{h_t}$ and denote the corresponding vector of conditional proportions following each of the histories in our chosen subset by \textbf{$\phi$}.  For every $q\in[0.5,1]$ we can calculate recursively the conditional probability $\phi(q)_{h_t}=\Pr(a_t=1|h_t,q).$ We denote the vector of conditional model probabilities by \textbf{$\hat\phi(q).$}  The GMM estimator $\hat{q},$ is the value for which the distance between the two  vectors is minimized, i.e,
$$
\hat{q}=argmin_{q\in[0.5,1]} ||\hat\phi(q)-\phi||.
$$

Due to technical reasons we chose to focus on all histories of length two\footnote{We chose to exclude longer histories as in several sessions, Anderson and Holt \cite{Anderson1997}  introduced a public signal after the third round. Additionally, by choosing only length two histories, we verify that the assumption of a weakly stationary ergodic stochastic process holds yielding a consistent  GMM estimator.}, i.e. $A=[00,01,10,11]$. In this approach we have 4 conditions and attempt to estimate a single parameter, thus the model is identified and can be estimated. As there are more conditions than parameters, we are in an  over-identified estimation case, and thus we use a two-step estimation to find the optimal weighting matrix.\footnote{The method of two-step efficient GMM estimation is described on Chapter 3 of \cite{Hayashi2000}. The Python code used for estimation, adapted from \cite{Evans2018}, is available available at \url{https://github.com/morankor/practical_inf_cascades/}.}

 The two-steps efficient GMM estimation resulted in $\hat q = 0.7171$ with a standard deviation of $0.0208.$ Slightly above (7.5\%) the actual signal quality of $q=\frac{2}{3}.$ Furthermore, note that  $q=\frac{2}{3}$ is inside the 99\% confidence interval  $CI_{99\%}=[0.6622,0.7718]$.

\section{Discussion}\label{sec:discussion}

The literature of social learning can be roughly divided into two groups, one in which the information structure is limited to a binary signal, and another in which the signal structure is assumed to be abstract. The lack of a method to generate richer, yet tractable signal structures poses a challenge for incorporating these models in applied theory or econometric work as the structure in the former group is too limited for the majority of complex environments or estimation methods, while the structure in the latter model group, while extremely useful for asymptotic analysis,  imposes difficulties when attempting to analyze finite horizon results or attempting to convey economic significance. 

In this paper  we attempt to bridge this gap by presenting a method of generating signal structures which are richer than the binary model of \cite{Banerjee1992,bikhchandani1992theory}, yet is more tractable than the abstract model of \cite{Smith2012}.  We demonstrate the advantage of our approach by revisiting two classical papers \cite{Smith2012,Anderson1997}. %\footnote{Amir: You can end the paragraph here. No need for details in the conclusion.} First, we calculate the probability of a contrary action for a finite sequence of actions, thus supplementing the seminal results of Smith and S\o rensen \cite{Smith2012}. Second, we revisit the experiment of Andersen and Holt \cite{Anderson1997} and successfully estimate the quality of signal they have used. 

Our goal in this paper is to provide new tools for applied theory and empirical research on social learning. Theoreticians can utilize our model to easily generate examples, perform numeric calculations, and complement their asymptotic results. Econometricians and experimenters can use our method to construct estimation methods which directly assess the effect of information cascades. In addition, we make a theoretical contribution by providing a necessary and sufficient condition for the occurrence of action cascades.

\newpage

% Bibliography
\bibliographystyle{plain}
\bibliography{crowdfunding}
\newpage
% Appendix
\appendix

% The appendix command is issued once, prior to all appendices, if any.

\section{An example of IHRP violation without action cascades.}\label{sec:ihrp_compact_ex}

Consider an alternative information structure. Assume that agent signals are drawn from $(1-c,c)$ for some $c\in[\frac{1}{2},1].$ In state $\omega=0,$ the signals are uniformly drawn, i.e for all $x\in(1-c,c),$ $g_0(x)=\frac{1}{2c-1},$ and for all $x\notin (1-c,c),$  $g_0(x)=0.$ 
In state $\omega=1$ the signals are drawn by,
$$
g_1(x)=\begin{cases}
\frac{1-\varepsilon}{2c-1}&\mbox{ if }x\in (1-c,\frac{1}{2})\\
\frac{1+\varepsilon}{2c-1}&\mbox{ if }x\in[\frac{1}{2},c)\\
0&\mbox{ otherwise}
\end{cases},
$$
for some $\varepsilon\in(0,1).$

Now, for some $\delta\in(0,1),$ assume that with probability $\delta,$  agent signals are generated by our example from Section \ref{sec: condition} (with the parameters $\bar{q}>c$ and $\ubar{q}=c$), and with probability $(1-\delta)$ signals are generated as described above.  Note that the resulting compound distributions have a compact domain, it exhibits the MLRP property (by construction), yet for sufficiently large $\varepsilon,$ it violates the IHRP property.

Using Python, we calculate the example for  $\delta=0.5,~c=0.56,~\bar{q}=0.8,$ and $\varepsilon=0.9999.$ In this example, IHRP is violated, yet action cascades never occur.

\end{document}